\documentclass{article}

\usepackage{amsthm}
\usepackage{amsmath}
\usepackage{amssymb}
\usepackage{amsfonts}
\usepackage{booktabs} 
\usepackage{graphicx}
\usepackage{hyperref}

\usepackage[margin=2cm]{geometry}

\usepackage{times}
\usepackage[authoryear]{natbib}
\usepackage{enumerate}

\usepackage{multirow}
\usepackage[flushleft]{threeparttable}
\usepackage{subfigure}
\newcommand\numberthis{\addtocounter{equation}{1}\tag{\theequation}}

\newtheorem{theorem}{Theorem}
\theoremstyle{plain}

\theoremstyle{definition}

\newtheorem{lemma}[theorem]{Lemma}

\DeclareMathOperator\Hy{\mathcal{H}}
\DeclareMathOperator\Var{Var}

\DeclareMathOperator\E{\mathbb{E}}
\DeclareMathOperator\Pow{Pow}
\DeclareMathOperator\pr{P}


\begin{document}

\title{Optimal design of the Wilcoxon-Mann-Whitney-test}
\author{Paul-Christian B\"urkner, Philipp Doebler, and Heinz Holling}

\date{Email: paul.buerkner@gmail.com}

\maketitle


\begin{abstract}
In scientific research, many hypotheses relate to the comparison of two independent groups. Usually, it is of interest to use a design (i.e., the allocation of sample sizes $m$ and $n$ for fixed $N = m + n$) that maximizes the power of the applied statistical test. It is known that the two-sample t-tests for homogeneous and heterogeneous variances may lose substantial power when variances are unequal but equally large samples are used. We demonstrate that this is not the case for the non-parametric Wilcoxon-Mann-Whitney-test, whose application in biometrical research fields is motivated by two examples from cancer research. We prove the optimality of the design $m = n$ in case of symmetric and identically shaped distributions using normal approximations and show that this design generally offers power only negligibly lower than the optimal design for a wide range of distributions. Please cite this paper as published in the Biometrical Journal (\url{https://doi.org/10.1002/bimj.201600022}). \\

\end{abstract}

\maketitle                   






\section{Introduction}

The comparison of two independent samples can be considered one of the most widespread applications of statistics, being used for instance in medicine, biology, neuroscience and psychology. For normally distributed data, t-tests are applied to check whether the difference between the samples is significant. When variances turn out to be equal across samples, the t-test for homogeneous variances ($T_{hom}$) is used. When variances are unequal, the t-test for heterogeneous variances ($T_{het}$; \citealp{Welch1938}) is the preferred method. In scientific experiments, researchers are free to choose how many subjects to allocate to the first and second sample respectively, while the total sample size $N$ is commonly limited due to time, money, or ethical restrictions. Almost always, researchers will use equally large sample sizes $m$ and $n$ for both groups, because they were taught doing it so or maybe just because `it feels right'. Indeed, $T_{hom}$ has the highest power if $m = n$, assuming normally distributed samples and equally large variances (e.g., \citealp{holling2013}). Less commonly known however, this is not the case for $T_{het}$. Here, the sample size should be higher for the sample with higher variance (again assuming normally distributed data), increasing proportionally with the ratio of the standard deviations \citep{dettemunk1997,dettebrien2004}. Furthermore, using non-optimal sample size allocations may result in substantial loss of power as compared to the optimum \citep{dettebrien2004}. Unfortunately, it is usually unclear a-priori, how much (if at all) variances will actually differ. For this reason, equally large sample sizes might be used nevertheless at the cost of potentially losing non-negligible amount of power. 

If the data are not normally distributed, t-tests may be invalid, especially for small samples. In this case, the Wilcoxon-Mann-Whitney-test proposed by \cite{wilcoxon1945} and \cite{mann1947} -- being arguably one of the most widely used non-parametric tests developed so far -- is a powerful alternative. It is applied in cancer research (e.g, \citealp{arap1998, sandler2003}), virology (e.g., \citealp{dibisceglie1989, misiani1994}), neuroscience (e.g.,      \citealp{shen2008}), and clinical psychology (e.g., \citealp{lanquillon2000}), to mention only a few areas of application and related studies. The study of \cite{sandler2003} and another cancer study by \cite{epping1994} are discussed in more detail in Section 2.2, to explain the usefulness of the Wilcoxon-Mann-Whitney-test for biometrical research. Generally, \citet[p. 1]{hollander2013} define a \emph{non-parametric} method as ``a statistical procedure that has certain desirable properties that hold under relatively mild assumptions regarding the underlying populations from which the data are obtained''. The key aspect of most non-parametric methods is that they are distribution free. That is they are valid for samples from a wide range of distributions and therefore no (or fewer) assumptions on the distribution of the data, in particular no normality assumption, have to be made (c.f., \citealp{agresti2013, hollander2013, sprent2007, brunner2002}).

Considerable amount of research has been conducted on the optimal designs of classical parametric methods such as t-tests (see above), linear and generalized linear models (e.g., see \citealp{berger2009}; \citealp{atkinson2007}). However, it appears that alternative non-parametric methods are less well represented in the optimal design literature. In particular, the optimal design of the Wilcoxon-Mann-Whitney-test has not yet been investigated so far, despite its broad area of application.

In the following, we introduce some notation that is used throughout the paper. Let $X$ and $Y$ be two independent random variables with distributions $F$ and $G$ from which we take $m$ and $n$ independent realizations. In case of two-sample tests, an exact (experimental) design is the allocation of sample sizes $m$ and $n$, while keeping $N = m + n$ constant. Define $m := \omega N$ as well as $n := (1-\omega)N$ for $\omega \in [0,1]$. Using $\omega$ instead of $m$ and $n$, every design can be symbolized by a number in $[0,1]$ independent of $N$. Of course, not all values $\omega \in [0,1]$ can be realized in practice, since $m$ and $n$ must be natural numbers, and hence $\omega$ is called an approximate design \citep{berger2009}. 

The null hypothesis considered in this paper is
\begin{equation}
\label{H0}
 \Hy_0: G(x) = F(x)
\end{equation} 
implying an identical distribution underlying both samples. In the most general case, the alternative hypothesis can be defined as
\begin{equation}
\label{H1}
 \Hy_1: G(x) \neq F(x),
\end{equation} 
although it is often useful and necessary to make the $\Hy_1$ more specific, in order to ensure good properties of the tests (see next section). The interpretation of a statistical tests' outcome depends heavily on the underlying pair of hypotheses. Accordingly, they should always be specified with care.

The \emph{power} of a test is the probability that $\Hy_0$ is rejected when $\Hy_1$ is actually true. In the present paper, an \emph{optimal design}, denoted as $\omega^*$, is a design that maximizes the power of a test for given $N$, $F$, and $G$. In case of t-tests, maximizing the power is equivalent to achieving D-optimality (cf. \citealp{berger2009}; \citealp{atkinson2007}), which is the most widely applied optimal design criterion. However, this equivalence does not hold for the Wilcoxon-Mann-Whitney-test discussed in this paper, so that optimality is defined by maximal power here.

As stated above, it is known that $T_{hom}$ has its optimal design exactly at $\omega^* = 0.5$ (e.g., \citealp{holling2013}), when $F$ and $G$ are normal with equal variances. That is, one should assign equal sample sizes to both samples. However, when $F$ and $G$ have unequal variances $\sigma_1^2$ and $\sigma_2^2$, $T_{het}$ has an locally optimal design which is close to $\omega^* = \frac{1}{1+\sigma_2/\sigma_1}$ \citep{dettemunk1997,dettebrien2004}. That is, the sample size should be higher for the sample with higher variance. 
It is now of interest to find the optimal design for the Wilcoxon-Mann-Whitney-test, which serves as a powerful non-parametric alternative to the classical t-tests.
 
\section{Wilcoxon-Mann-Whitney-test}

The test statistic $U_{mn}$ of the Wilcoxon-Mann-Whitney-test (in the following abbreviated as $T_U$) is defined as
\begin{equation}
  U_{mn} := \sum_{i=1}^m{\sum_{j=1}^n{\chi(x_i,y_j)}}
\end{equation}
with
\begin{equation}
 \chi(x_i,y_j) := \Bigg\{  \begin{array}{c} 
1 \text{ if } x_i \geq y_j \\
0 \text{ if } x_i < y_j.\\
\end{array}
\end{equation}

Under $\Hy_0$, the exact distribution of $U_{mn}$ is known and can be calculated using a recursive formula \citep{mann1947}. This recursive formula further allows to calculate the central moments of $U_{mn}$. The mean and the variance for continuous $F$ and $G$ under $\Hy_0$ are given by
\begin{equation}
\label{MVU0}
  \E_0(U_{mn}) = \frac{mn}{2} \qquad \Var_0(U_{mn}) = \frac{mn(m+n+1)}{12}.
\end{equation} 
The original $\Hy_1$ for $T_U$, often called stochastic ordering hypothesis \citep{fay2010}, states that one of the two random variables is stochastically larger than the other assuming the overall shape of the distributions to be the same. That is there is an $a \neq 0$ such that
\begin{equation}
\label{SH1}
 \Hy_1: G(x) = F(x + a).
\end{equation}
When using the above $\Hy_1$, $T_U$ is consistent \citep{mann1947} as well as unbiased for any one-sided hypothesis (i.e., $a > 0$ or $a <0$; \citealp{lehmann1951}; \citealp{van1950}). Unbiased means that the test is less likely to reject the $\Hy_0$ when it is true than when any other hypothesis is true. This property may not hold in the two-sided case \citep{van1950}. If samples are normal with equal variances, $T_{hom}$ is uniformly most powerful among all unbiased tests (e.g., \citealp{lehmann2006}). For this case of homogeneous variances, \cite{mood1954} has shown that the asymptotic efficiency of $T_U$ relative to $T_{hom}$ is $3/\pi \approx 0.955$. This is quite large given that $T_U$ has higher power than $T_{hom}$ in many non-normal situations, even for $N \rightarrow \infty$ \citep{hodges1956, blair1980, sawilowsky1992, fay2010}. 

When $F \neq G$, we have no recursion available to determine the distribution of $U_{mn}$, but it is nevertheless possible to calculate general formulae for the mean and the variance of $U_{mn}$. In the following, define $\pr(X \geq Y)$ as the probability of $X$ exceeding $Y$, that is the probability that some realization $x$ of $X$ is greater or equal to some realization $y$ of $Y$. Furthermore, define $\tilde{F}(x) := \pr(X < x)$ (note that $\tilde{F} = F$ if $F$ is continuous). \\

\begin{lemma} Let $F$ and $G$ be some arbitrary distributions with densities $f$ and $g$.
\label{MVU}
\begin{enumerate}[(i)]
\item We have
	\begin{equation*}
	\pr(X \geq Y) = \int{G(x)f(x)dx}
	\end{equation*}
where 

\item The mean and the variance of $U_{mn}$ can be written as
\begin{align*}
    \E(U_{mn}) &= mn \pr(X \geq Y) = \omega (1-\omega)N^2 \pr(X \geq Y), \\
    \Var(U_{mn}) &= mn \bigg( \pr(X \geq Y) -(m+n-1)\pr(X \geq Y)^2 \\
		&\quad +(n-1)\int{G(x)^2f(x)dx} +(m-1)\int{(1-\tilde{F}(x))^2g(x)dx} \bigg). \\
		&= \omega (1-\omega)N^2 \bigg( \pr(X \geq Y) -(N-1)\pr(X \geq Y)^2 \\
		&\quad +((1-\omega) N -1)\int{G(x)^2f(x)dx} +(\omega N -1)\int{(1-\tilde{F}(x))^2g(x)dx} \bigg). 
	\end{align*}
\item If $F$ and $G$ are symmetric with $G(x) = F(x+a)$ for $a \in \mathbb{R}$ it holds that
\begin{equation*}
  \int{G(x)^2f(x)dx} = \int{(1-\tilde{F}(x))^2g(x)dx}.
\end{equation*}	
\end{enumerate}
\end{lemma} 
The proofs of Lemma \ref{MVU} $(i)$ and $(ii)$ can be found in the Appendix of \cite{lehmann2006np}. The full version of Lemma \ref{MVU} as well as the remaining proofs can be found in Appendix A of the present paper. The asymptotic normality of $U_{mn}$ under $\Hy_0$ was originally proven by \cite{mann1947}. Expanding this result, \cite{lehmann1951} used a theorem of \cite{hoeffding1948} to prove the general asymptotic normality (holding also under $\Hy_1$) for a large class of estimators, including $U_{mn}$, under the assumption
\begin{equation}
\label{cons_res}
 m/n = \text{constant } \text{as } N \rightarrow \infty.
\end{equation}  

\subsection{Optimal design of $T_U$ for the stochastic ordering hypothesis}
 
If we assume the stochastic ordering hypothesis (\ref{SH1}) to be true and focus on \emph{symmetric} distributions only, we are able to find the optimal design of $T_U$ analytically at least for larger sample sizes. \\

\begin{theorem}
\label{theorem_sym}
  Consider all designs $\omega \in [\varepsilon, 1 - \varepsilon]$ for any fixed $\varepsilon \in (0,0.5)$ and let $N$ be sufficiently large so that $U_{mn}$ is approximately normal for all those designs. Then, for symmetric continuous distributions $F$ and $G$ with $G(x) = F(x + a)$ for some $a \neq 0$, the optimal design is given if $\omega^* = 0.5$.
\end{theorem}

\begin{proof}
 We write $U_N(\omega)$ instead of $U_{mn}$ to make the dependency on $\omega$ explicit. Transform $U_N(\omega)$ so that it has mean zero and variance one under $\Hy_0$. Applying the same transformation to $U_N(\omega)$ under $\Hy_1$, we arrive at
\begin{equation}
\label{mu_sigma}
 \mu_N(\omega) := \frac{\E(U_N(\omega))-\E_0(U_N(\omega))}{\sqrt{\Var_0(U_N(\omega))}} \text{ and } \sigma_N^2(\omega) := \frac{\Var(U_N(\omega))}{\Var_0(U_N(\omega))}
\end{equation} 
as the transformed mean and variance, respectively. As $F$ and $G$ are symmetric and continuous with $G(x) = F(x+a)$ for some $a \neq 0$, then from Lemma \ref{MVU} $(iii)$ and the definitions in (\ref{mu_sigma}) we conclude
\begin{equation}
 \mu_N(\omega) = \frac{\omega(1-\omega)N^2(\pr(X \geq Y)-1/2)}  
  {\sqrt{\omega(1-\omega)N^2(N+1)/12}}
  = \frac{\sqrt{\omega(1-\omega)}N(\pr(X \geq Y)-1/2)}  
  {\sqrt{(N+1)/12}}
\end{equation}
and
\begin{align*}
\label{sigma_simple}
\sigma^2_N(\omega) &= \frac{\omega(1-\omega)N^2(\pr(X \geq Y)-(N-1)\pr(X \geq Y)^2 + (N-2)\int{G(x)^2f(x)dx})}{\omega(1-\omega)N^2(N+1)/12}  \\
&= \frac{\pr(X \geq Y)-(N-1)\pr(X \geq Y)^2 + (N-2)\int{G(x)^2f(x)dx}}{(N+1)/12}. \numberthis
\end{align*}
We see that $\sigma^2_N$ is independent of $\omega$ and $\mu_N$ only depends on $\omega$ through $\sqrt{\omega(1-\omega)}$, which is maximized at $\omega=0.5$. Thus, $\omega^* = 0.5$.
\end{proof}


Next we investigate the deficiency $D$, which corresponds to the percentage of additional sample size needed so that a given design offers as much power as the optimal design. We see from (\ref{sigma_simple}) that for larger $N$ the variance $\sigma^2_N$ is approximately independent of $N$. Thus, two designs have approximately the same power if their corresponding means $\mu_N(\omega)$ are identical. For larger $N$, this leads to the equation
\begin{equation}
 \sqrt{\omega(1-\omega)N(1+D)} = \sqrt{\frac{N}{4}}
\end{equation}
with solution 
\begin{equation}
  D(\omega) = \frac{1}{4\omega(1-\omega)} - 1.
\end{equation}
The graph of the deficiency is displayed in Figure~\ref{fig:deficiency}. Interestingly, the same deficiency formula can be found for $T_{hom}$ (\citealp{holling2013}).

The independence of $\sigma^2_N(\omega)$ on $\omega$ will generally not hold if $F$ and $G$ are asymmetric or if the stochastic ordering hypothesis is not satisfied, that is if $F$ and $G$ differ in their overall shapes. These cases are considered in the next section.  
 
\subsection{Optimal design of $T_U$ for the general alternative hypothesis}
\label{section_general_hyp}

In the following, we will use the general $\Hy_1$ from equation (\ref{H1}). Although $T_U$ was originally defined only under the stochastic ordering hypothesis (\ref{SH1}) \citep{mann1947, fay2010}, many practically relevant cases, such as the comparison of two normal distributions with unequal variances or differently skewed distributions fall outside its scope. The latter appears quite frequently in biometrical research fields, for instance, in cancer research and we want to present two case examples here. 

\cite{sandler2003} investigated the effect of regular aspirin to prevent the recurrence of colorectal cancer by comparing a group receiving a small daily doses of aspirin to an equally large placebo control group. Among others, they analyzed the number of adenomas detected after a certain time period as well as the recurrence time   of adenomas. Both types of variables are skewed if values are close to zero, because they are naturally bound there. The distribution of the number of detected adenomas in both groups were reported by \cite{sandler2003}. Most subjects did not have any adenomas ($83\%$ in the experimental group vs. $73\%$ in the control group), some had one ($10\%$ vs. $14\%$) or two ($3\%$ vs. $7\%$) and even fewer and three or more ($3\%$ vs. $5\%$). It is immediately evident that the distributions are highly skewed so that the application of t-tests might be at least questionable. \cite{sandler2003} thus compared the number of adenomas in both groups using $T_U$ and found that the experimental group had significantly fewer adenomas. 

In another study, \cite{epping1994} investigated the effect of cognitive avoidance on cancer recovery, by comparing the group of patients successfully recovering from cancer to those who failed to recover. The distribution of cognitive avoidance in both groups is displayed in Figure~\ref{fig:avoidance}. While the distribution of patients not recovering is relatively symmetric, it is clearly right skewed for patients who did recover. Among others, the optimal design of $T_U$ for this special case is examined below. Again, the application of t-tests is questionable here and $T_U$ may be a sensible alternative. These examples show the relevance of the Wilcoxon-Mann-Whitney-test in biometrical applications. Thus, it is of interest to investigate its optimal design when distributions are known to be skewed and / or have unequal variances. 

Using the normal approximation, the power of $T_U$ can be approximated as well. In the one-sided case, if the alternative hypothesis is 
$\pr(X \geq Y) > 0.5$ (equivalent to $a > 0$ under $\Hy_1$ (\ref{SH1}) in case of continuous $X$ and $Y$), we have
\begin{equation}
\label{pow.one}
  \Pow_N(\omega) = 1 - \Phi\left(\frac{z_{1-\alpha}-\mu_N(\omega)}{ \sigma_N(\omega)} \right).
\end{equation}
In the two-sided case $\pr(X \geq Y) \neq 0.5$ (equivalent to $a \neq 0$ under $\Hy_1$ (\ref{SH1}) in case of continuous $X$ and $Y$) we have
\begin{equation}
\label{pow.two}
  \Pow_N(\omega) = \Phi\left(\frac{z_{\alpha/2} - \mu_N(\omega)}{ \sigma_N(\omega)} \right) - \Phi\left(\frac{z_{1-\alpha/2} - \mu_N(\omega)}{\sigma_N(\omega)} \right) + 1.
\end{equation}
Here, $\alpha$ denotes the nominal $\alpha$-level of the test, $\Phi$ the standard normal distribution function, and $z_{\alpha}$ the $\alpha$ quantile of the standard normal distribution. Unfortunately, finding the maximum of $\Pow_N$ in $\omega \in [0,1]$ analytically turns out to be unfeasible. For this reason, we present some examples for common distributions below and focus on the one-sided version of $T_U$, as the two-sided case yields nothing new concerning the optimal design despite a reduced power in general.

One of the most common cases of two-sample comparisons occurring for real data, which is not covered by Theorem~\ref{theorem_sym}, is the comparison of two normal distributions with unequal variances. The optimal design of $T_U$ was investigated for variance ratios ranging from $1/9$ to $9$ (or equivalently $1/3$ to $3$ in terms of standard deviation ratios), as the vast majority of real data comparisons can be expected to fall within this interval. Figure~\ref{fig:dens} (a) provides an overview on the applied normal densities. Recall that, in case of two normal samples with unequal variances, the optimal design of $T_{het}$ assigns a higher sample size to the sample with higher variance. Also, one may lose substantial power when variances are unequal but equally large samples are used \citep{dettebrien2004}. Using the asymptotic power function (\ref{pow.two}), a different pattern can be found for $T_U$ (see Figure~\ref{fig:power_TU} (a) and (b) and (c)): When $F$ and $G$ are normal with $\sigma_1^2 \neq \sigma_2^2$, the optimal design $\omega^*$, quite surprisingly, does not always favor the sample with the higher variance. Instead, the pattern appears to be more complex as displayed in Figure~\ref{fig:power_TU} (a). Furthermore, $\omega^*$ depends on the total sample size $N$ (see Figure~\ref{fig:power_TU} (b)) in a way that higher values of $N$ lead to slightly higher values of $\omega^*$ at least for the presented examples. Also, the optimal design depends very slightly on the $\alpha$-level (see Figure~\ref{fig:power_TU} (c)). To demonstrate that these findings are not just artifacts of the approximations, Figure~\ref{fig:power_TU} contains the simulated power based on 10,000 trials (dashed lines) along with the approximated power. From our perspective, the most important observation is that one generally loses only little power when choosing $\omega = 0.5$ as compared to the optimal design $w^*$. To illustrate that this property may indeed not hold for $T_{het}$ when variances differ considerably (e.g., for a standard deviation ratio of $3$), its simulated power is also displayed in Figure~\ref{fig:power_TU} (a), (b), and (c) as grey lines.

Often enough, real data are not normally distributed, but skewed in some way as shown for the cancer studies of \cite{sandler2003} and \cite{epping1994}. Even though the central limit theorem ensures normally distributed means for $N \rightarrow \infty$ and thus the validity of the two-sample t-tests, sample sizes may be too small in many cases to ensure sufficient convergence to normality. Furthermore, even for larger (or infinitely large) samples, $T_U$ will be more efficient than the t-tests for many skewed distributions \citep{hodges1956, blair1980, sawilowsky1992}. In fact, this can be considered one of the main reasons to apply the Wilcoxon-Mann-Whitney-test. Accordingly, the optimal design of $T_U$ for skewed distributions is of primary interest. In the absence of any general analytic solution for this case, we will again investigate the optimal design for selected examples. We chose to use exponential, log-normal and $\chi^2$-distributions to represent different shapes and amounts of skewness that are typically present in real data (e.g., for reaction or survival times). We decided to include only unimodal distributions, as multimodal distributions appear pretty rarely and usually indicate that different populations, each having a unimodal distributions, were mixed up (something that should be avoided to allow clear interpretation of scientific results). First, consider the canonical case of skewed distributions with the same shape but shifted mean, which, by definition, satisfy the stochastic ordering hypothesis (\ref{SH1}). See Figure~\ref{fig:dens} (b) for a visualization of their densities. As can be seen from the power functions displayed in Figure~\ref{fig:power_TU} (d) and (e), $\omega^*$ varies with the degree of the shift, the total sample size, and (slightly) with the amount of skewness. Again however, $\omega = 0.5$ is generally nearly as good as $w^*$. Leaving the stochastic ordering hypothesis, distributions with different amount of skewness were also compared to each other (densities are displayed in Figure~\ref{fig:dens} (c) and (d)). From the power functions in Figure~\ref{fig:power_TU} (f), (g), and (h), it is immediately evident that $\omega = 0.5$ is again nearly optimal for all displayed comparisons.

To end this section, we want to briefly and exemplary discuss the optimal design of $T_U$ for the cancer study of \cite{epping1994} mentioned above. Overall, $67$ patients participated in the study. The distribution of patients successfully recovering from cancer can be nicely approximated by a $\chi^2$ distribution with $14$ degrees of freedom, whereas the distribution of patients who failed to recover is similar to a Student-t distribution with $3$ degrees of freedom, location parameter of $17$, and scale parameter of $2.8$ (see Figure~\ref{fig:example_approx_dens} for an illustration of the corresponding densities). When applying $T_U$ under these conditions to test whether the first group has smaller values than the second, we find the optimal design to be roughly $\omega^* \approx 0.45$ (see Figure \ref{fig:example_power}) that is around $67 \times 0.45 \approx 30$ people should be assigned to the first group. Note that again, almost no power is lost when using equally large groups.

\section{Conclusion}

In the present paper, we demonstrated that assigning equally large sample sizes to both groups is generally a very good design for the Wilcoxon-Mann-Whitney-test. Using normal approximations, its optimality was proved for symmetric and identically shaped distributions. For a range of other distributions, we could show that $\omega = 0.5$ offers power only negligibly lower than the optimal design that varies with the underlying distributions, the total sample size and the $\alpha$-level. Thus, in contrast to the t-tests for homogeneous and heterogeneous variances, equally large sample sizes can be used without the risk of substantially losing power. In this sense, the Wilcoxon-Mann-Whitney-test is not only valid and powerful for a wide range of distributions because of its non-parametric nature, but also offers a robust experimental design to be applied on.

The examples presented in the figures are only a subset of the cases we analyzed in order to come to our conclusions. However, since $F$ and $G$ come from an infinite space of distributions, there might still be relevant combinations, for which the design $\omega = 0.5$ has considerably lower power than the optimal design, even for medium to large samples sizes. Ideally, we also wanted to provide a closed form of the optimal design for arbitrary $F$ and $G$. However, even when the (asymptotic) distribution under $\Hy_1$ is known, it is practically unfeasible to solve for $\omega^*$ if $F$ and $G$ are not assumed to be symmetric and identically shaped. If one has a clear idea on how the data will be distributed and wants to get a more precise estimation of the optimal design (instead of just applying $\omega = 0.5$), numerical optimization of the power function still appears to be the best option. \\



\section*{Appendix}

\subsection*{A: Extension and proof of Lemma \ref{MVU}}

Let $F$ and $G$ be some arbitrary (discrete or continuous) distributions with densities $f$ and $g$.
\begin{enumerate}[(i)]
\item We have
	\begin{equation*}
	\pr(X \geq Y) = \int{G(x)f(x)dx}.
	\end{equation*}
	
\item The mean and the variance of $U_{mn}$ can be written as
\begin{align*}
    \E(U_{mn}) &= mn \pr(X \geq Y), \\
    \Var(U_{mn})  &= nm \bigg( \pr(X \geq Y) -(n+m-1)\pr(X \geq Y)^2 \\
		&\quad +(n-1)\int{G(x)^2f(x)dx} +(m-1)\int{(1-\tilde{F}(x))^2g(x)dx} \bigg).  
	\end{align*}
	
\item If $F$ and $G$ are symmetric with $G(x) = F(x+a)$ for fixed $a \in \mathbb{R}$ it holds that
\begin{equation*}
  \int{G(x)^2f(x)dx} = \int{(1-\tilde{F}(x))^2g(x)dx}.
\end{equation*}

\item (supplemental) For continuous $F$ and $G$ under the null hypothesis $F = G$, it holds that
	\begin{equation*}
		\E(U_{mn})  = \E_0(U_{mn}) \text{ and } \Var(U_{mn}) = \Var_0(U_{mn}) 
	\end{equation*}
	that is the formulae above coincide with (\ref{MVU0}).
	
\item (supplemental) For continuous $F$ and $G$ we have
	\begin{equation*}
		\int{(1-F(x))^2g(x)dx} = \int{\left(\int_{-\infty}^{x}{G(y)f(y)dy} + G(x)(1-F(x))\right)f(x)dx}.
	\end{equation*}
\end{enumerate}

\begin{proof}
$(i)$ and $(ii)$ See the Appendix of \cite{lehmann2006np}. \\
   
$(iii)$ Without loss of generality assume $\E(X) = 0$ so that $\E(Y) = -a$ because $G(x) = F(x+a)$. Due to symmetry of $F$ and $G$ and since $G(x) = F(x+a)$, we have $f(x) = g(-a -x)$ and $G(x) = 1-\tilde{F}(-a -x)$ for all $x \in \mathbb{R}$, which also means 
\begin{equation}
  f(x)G(x)^2 = g(-a-x)(1-\tilde{F}(-a-x))^2.
\end{equation}
In both integrals, integration is done over $\mathbb{R}$ so that the statement follows immediately.

$(iv)$ The statement is trivial for $\E(U_{mn})$. For $F = G$ we have $\pr(X \geq Y) = 0.5$ and hence
\begin{equation}
\label{var.equal}
  \int{g(x)(1-F(x))^2dx} = 1 - 2\pr(X \geq Y) + \int{g(x)F(x)^2dx} = \int{f(x)G(x)^2dx}.
\end{equation}
As per definition of the mean, we may write 
\begin{equation}
 \int{f(x)F(x)^2dx} = \E(F(X)^2).
\end{equation}
The distribution of the random variable $F(X)$ is known to be uniform in $[0,1]$, if $X$ has distribution $F$. Accordingly, it holds that
\begin{equation}
  \E(F(X)^a) = \int_{0}^{1}{x^a dx} = \frac{1}{a+1}.
\end{equation}
In particular for $F = G$, this means
\begin{equation}
 \int{f(x)G(x)^2dx} = \int{g(x)(1-F(x))^2dx} = \frac{1}{3},
\end{equation}
which directly leads to
\begin{align*}
 \Var(U_{mn}) &= mn \bigg( \frac{1}{2} + \frac{m+n-2}{3} - \frac{m+n-1}{4} \bigg) \\
&= \frac{mn(6+4(m+n-2)-3(m+n-1))}{12} \\
&= \frac{mn(m+n+1)}{12}.
\numberthis
\end{align*}
Note that this statement does not hold when $F = G$ is discrete since $\pr(X \geq Y) \neq 0.5$ in this case due to the presence of ties. \\
$(v)$ The proof relies on another derivation of a formula for $\Var(U_{mn})$: For $t \in \mathbb{N}$, $t \geq \max(m,n)$, we may write every sequence of $x_1,...,x_m$, $y_1,...,y_n$ ordered by size  as 
\begin{equation}
 X_1, Y_1, X_2, ..., X_t, Y_t 
\end{equation} 
with $X_i \in \{0,...,m\}$ representing $X_i$ values out of $x_i,...,x_m$ and $Y_i \in \{0,...,n\}$ representing $Y_i$ values out of $y_i,...,y_n$. For instance, suppose that $t = m = n = 3$ and that ordering the observations $x_i,y_i$ ($1 \leq i \leq 3$) have resulted in
\begin{equation}
\label{ex_5}
  x_1,x_2,y_1,y_2,x_3,y_3.
\end{equation}
Then, we could represent sequence (\ref{ex_5}) as
\begin{equation}
  2(=X_1), 2(=Y_1), 1(=X_2), 1(=Y_2), 0(=X_3), 0(=Y_3).
\end{equation}

 Note that $\sum_{i=1}^t{X_i} = m$ 
and $\sum_{i=1}^t{Y_i} = n$. We can now write $U_{mn}$ as
\begin{equation}
  U_{mn} = \sum_{i=2}^t{\sum_{j=1}^{i-1}{X_i Y_j }}.
\end{equation}
Assume that both $(X_1,...,X_t)$ and $(Y_1,...,Y_t)$ are multinomial distributed with probabilities $(f_1,...,f_t)$ and $(g_1,...,g_t)$ respectively. For every $X_i$, it holds that  
\begin{equation}
  \E(X_i) = mf_i \text{ and } \E(X_i^2) = mf_i(1-f_i+mf_i).
\end{equation}
For $i \neq k$ we have
\begin{equation}
  \E(X_iX_k) = m(m-1)f_if_k.
\end{equation}
The same goes, of course, for $Y_i$. \\

Defining $\widetilde{f}_{i,m} := f_i(1-f_i-mf_i)$ and  $\widetilde{g}_{i,n} := g_i(1-g_i-ng_i)$, the second non-central moment $\E(U_{mn}^2)$ equals
\begin{align*}
\label{MN}
	\E(U_{mn}^2) &= \E((\sum_{i=2}^t{\sum_{j=1}^{i-1}{X_i Y_j}})^2) \\
	&= \sum_{i=2}^t{\sum_{j=1}^{j-1}{\E(X_i^2)\E(Y_j^2)}} + 
		 \sum_{i=3}^t{\sum_{j,l < i \atop j \neq l}{\E(X_i^2)\E(Y_j Y_l)}} \\
	&\quad +	 \sum_{i, k \leq t \atop i \neq k}{\sum_{j < i, k}{\E(X_i X_k)\E(Y_j^2)}} 
	+ \sum_{i,k \leq t \atop i \neq k}{\sum_{j < i; \text{ } l < k \atop j \neq l}{\E(X_i X_k)\E(Y_j Y_l)}} \\
	&= mn \bigg( \sum_{i=2}^t{\sum_{j=1}^{i-1}{\widetilde{f}_{i,m} \widetilde{g}_{j,n} }} + 
	   (n-1) \sum_{i=3}^t{\sum_{j,l < i \atop j \neq l}{ \widetilde{f}_{i,m} g_j g_l }} \\
	&\quad +	(m-1) \sum_{i, k \leq t \atop i \neq k}{\sum_{j < i, k} { f_i f_k \widetilde{g}_{j,n} }}
  +  (m-1)(n-1)\sum_{i,k \leq t \atop i \neq k}{\sum_{j < i; \text{ } l < k \atop j \neq l}{ f_i f_k g_j g_l }} \bigg) \\
	&= mn \bigg( \sum_{i=2}^t{\sum_{j=1}^{i-1}{\widetilde{f}_{i,m} \widetilde{g}_{j,n} }} + 
	   (n-1) (\sum_{i=3}^t{\widetilde{f}_{i,m} ((\sum_{j < i}{g_j})^2 - \sum_{j < i}{ g_j^2 })}) \\
	&\quad +	 (m-1) (\sum_{i, k \leq t}{\sum_{j < i, k} { f_i f_k \widetilde{g}_{j,n} }} - \sum_{i \leq t}{\sum_{j < i} { f_i^2 \widetilde{g}_{j,n} }}) \\
	&\quad + (m-1)(n-1)(\sum_{i,k \leq t}{f_i f_k \sum_{j < i}{g_j} \sum_{l < k}{g_l}} - \sum_{i \leq t}{f_i^2 (\sum_{j < i}{g_j})^2} - \sum_{i,k \leq t}{\sum_{j < i,k}{ f_i f_k g_j^2}}) \bigg). \numberthis 
\end{align*}	
For non-empty finite sets $A_t, B_t$ with $|A_t| = |B_t| = t$ define 
\begin{equation}
\label{def_pS}
f_{A_t}(x) :=  \Bigg\{  \begin{array}{l} 
\frac{f(x)}{\sum_{x \in A_t}{f(x)}} \text{ if } x \in A_t \\
0 \qquad \qquad \text{else } \\
\end{array}  \quad
g_{B_t}(y) :=  \Bigg\{  \begin{array}{l} 
\frac{g(y)}{\sum_{y \in B_t}{g(y)}} \text{ if } y \in B_t \\
0 \qquad \qquad \text{else.} \\
\end{array} 
\end{equation}
We choose $A_t$ and $B_t$ so that  $x_i < y_i < x_{i+1} < y_{i+1}$ for all $x_i \in A_t$ and $y_i \in B_t$ as well as 
\begin{equation}
  \lim_{t \rightarrow \infty} \sum_{x_i \leq x}{f_{A_t}(x_i)} = F(x) \qquad \lim_{t \rightarrow \infty} \sum_{y_i \leq y}{g_{B_t}(y_i)} = G(y).
\end{equation}
Since $F$ and $G$ are continuous, we have
\begin{equation}
 \lim_{t \rightarrow \infty} f_{A_t}(x_i) = \lim_{t \rightarrow \infty} g_{B_t}(y_j) = 0. 
\end{equation}
Together with $\sum_{i=1}^t{f_{A_t}(x_i)} = \sum_{i=1}^t{g_{B_t}(y_i)} = 1$, it follows that
\begin{equation}
\label{sump0}
\lim_{t \rightarrow \infty} \sum_{i=1}^t{f_{A_t}(x_i)^2} = \lim_{t \rightarrow \infty} \sum_{i=1}^t{g_{B_t}(y_i)^2} = 0. 
\end{equation}
Setting $f_i := f_{A_t}(x_i)$, $g_i := g_{B_t}(y_i)$ in (\ref{MN}) and applying (\ref{sump0}) leads to
\begin{align*}
\E(U_{mn}^2) &= mn \lim_{t \rightarrow \infty} \bigg( \sum_{i=2}^t{f_{A_t}(x_i) \sum_{j < i}{g_{B_t}(y_j)}} + (n-1) \sum_{i=3}^t{f_{A_t}(x_i) (\sum_{j < i}{g_{B_t}(y_j)})^2} \\
 &\quad +  (m-1)\sum_{i, k \leq t}{f_{A_t}(x_i) f_{A_t}(x_k) \sum_{j < i, k} {{g_{B_t}(y_j)} }} \\
 &\quad + (m-1)(n-1)(\sum_{i,k \leq t}{f_{A_t}(x_i) f_{A_t}(x_k) \sum_{j < i}{g_{B_t}(y_j)} \sum_{l < k}{g_{B_t}(y_l)}} \bigg) \\  
&= mn \lim_{t \rightarrow \infty} \bigg( \sum_{i=2}^t{f_{A_t}(x_i) \sum_{j < i}{g_{B_t}(y_j)}} + (n-1) \sum_{i=3}^t{f_{A_t}(x_i) (\sum_{j < i}{g_{B_t}(y_j)})^2} \\
 &\quad +  (m-1)\sum_{i \leq t}{f_{A_t}(x_i) \sum_{k \leq t} f_{A_t}(x_k) \sum_{j < i, k} {{g_{B_t}(y_j)} }} \\
 &\quad + (m-1)(n-1) \sum_{i \leq t}f_{A_t}(x_i) \sum_{j < i}{g_{B_t}(y_j)} \sum_{k \leq t} f_{A_t}(x_k)  \sum_{l < k}g_{B_t}(y_l) \bigg). \numberthis
\end{align*}
Going from sums to integrals and simplifying yields
\begin{align*}
	\E(U_{mn}^2) &= mn \bigg( \int{f(x)G(x)dx} + (n-1)\int{f(x)G(x)^2dx} \\
	&\quad +(m-1) \int f(x) \int f(y) \, G(\min(x,y)) dy dx 	\\
	&\quad + (m-1)(n-1) \left(\int{f(x)G(x)dx}\right)^2 \bigg)\\
&= mn \bigg( \pr(X \geq Y) + (n-1)\int{f(x)G(x)^2dx} \\
	&\quad +(m-1) \int f(x) \left( \int_{-\infty}^{x} f(y)G(y)dy + G(x)(1-F(x)) \right) dx  \\
	&\quad -(m+n-1)\pr(X \geq Y)^2 \bigg) + \E(U_{mn})^2.	
	\numberthis
\end{align*}	
Since $\Var(U_{mn}) = \E(U_{mn}^2) - \E(U_{mn})^2$ the statement follows by comparing this variance formula to the formula in $(ii)$. 
\end{proof}


\bibliography{bib}{}
\bibliographystyle{apalike}

\begin{figure}[hbtp] 
  \centering
  \includegraphics[height=0.50\textheight]{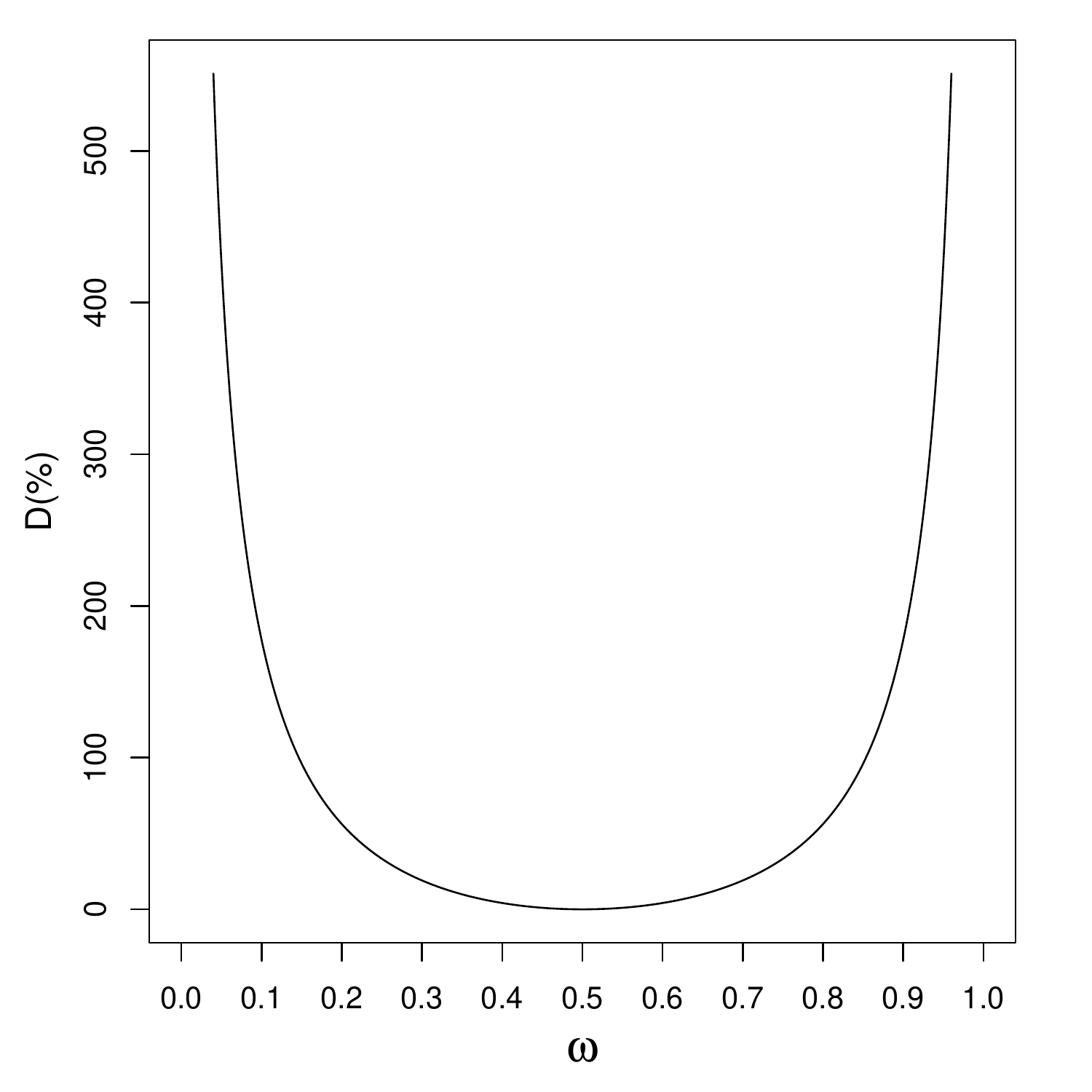}
  \caption{Graph of the approximate deficiency of the Wilcoxon-Mann-Whitney-test for symmetric and identically shaped distributions.}
  \label{fig:deficiency}
\end{figure}

\begin{figure}[hbtp] 
  \centering
  \includegraphics[height=0.50\textheight]{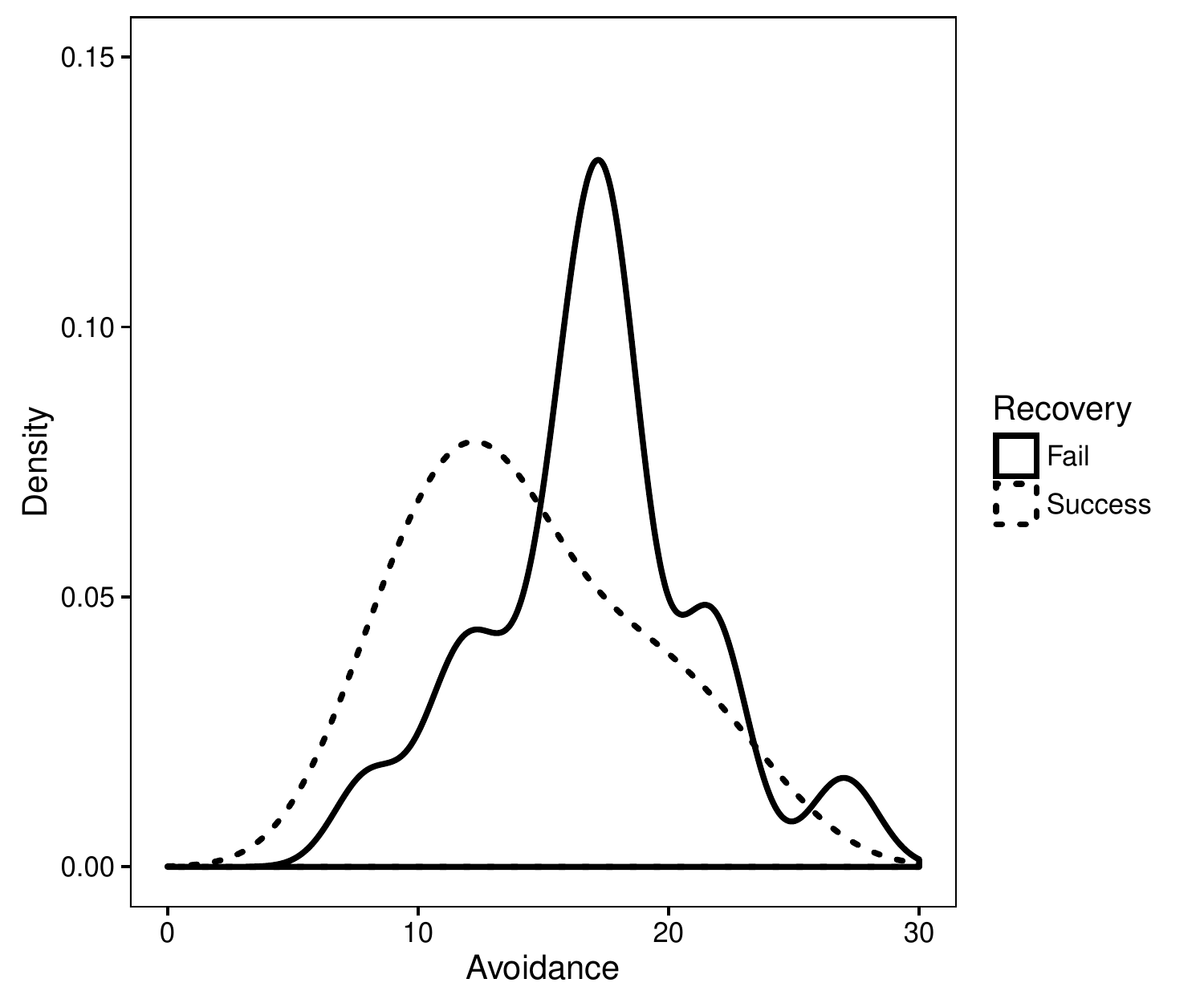}
  \caption{Smoothed densities of the distribution of cognitive avoidance for patients who managed to recover from cancer (Success) and for patients who did not recover (Fail) in the study of \cite{epping1994}. Details are provided in Section~\ref{section_general_hyp}.}
  \label{fig:avoidance}
\end{figure}

\begin{figure}[hbtp] 
  \subfigure[Normal densities corresponding to Subfigure (a), (b) and (c) in Figure \ref{fig:power_TU}.]{\includegraphics[width=0.49\textwidth]{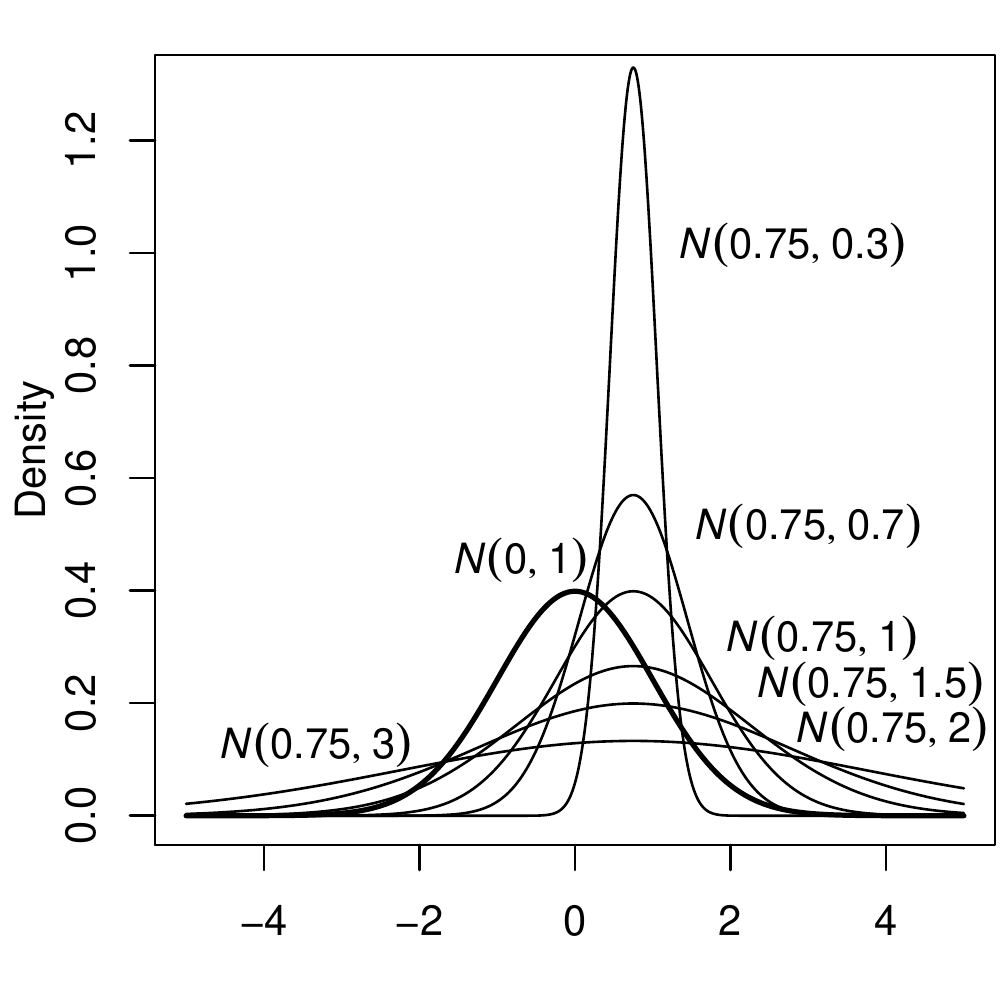}}
  \subfigure[Shifted skewed densities corresponding to Subfigure (d) and (e) in Figure \ref{fig:power_TU}.]{\includegraphics[width=0.49\textwidth]{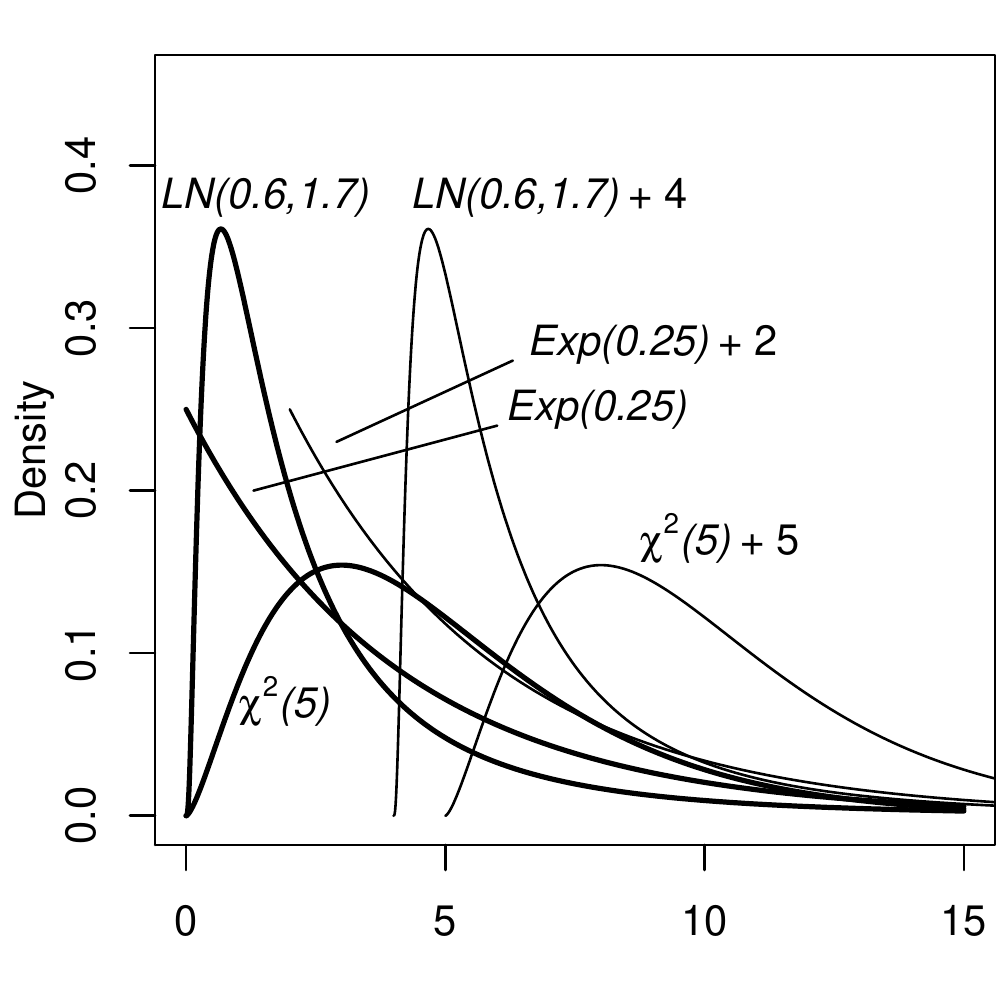}}
  \subfigure[Skewed densities corresponding to Subfigure (f) in Figure \ref{fig:power_TU}.]{\includegraphics[width=0.49\textwidth]{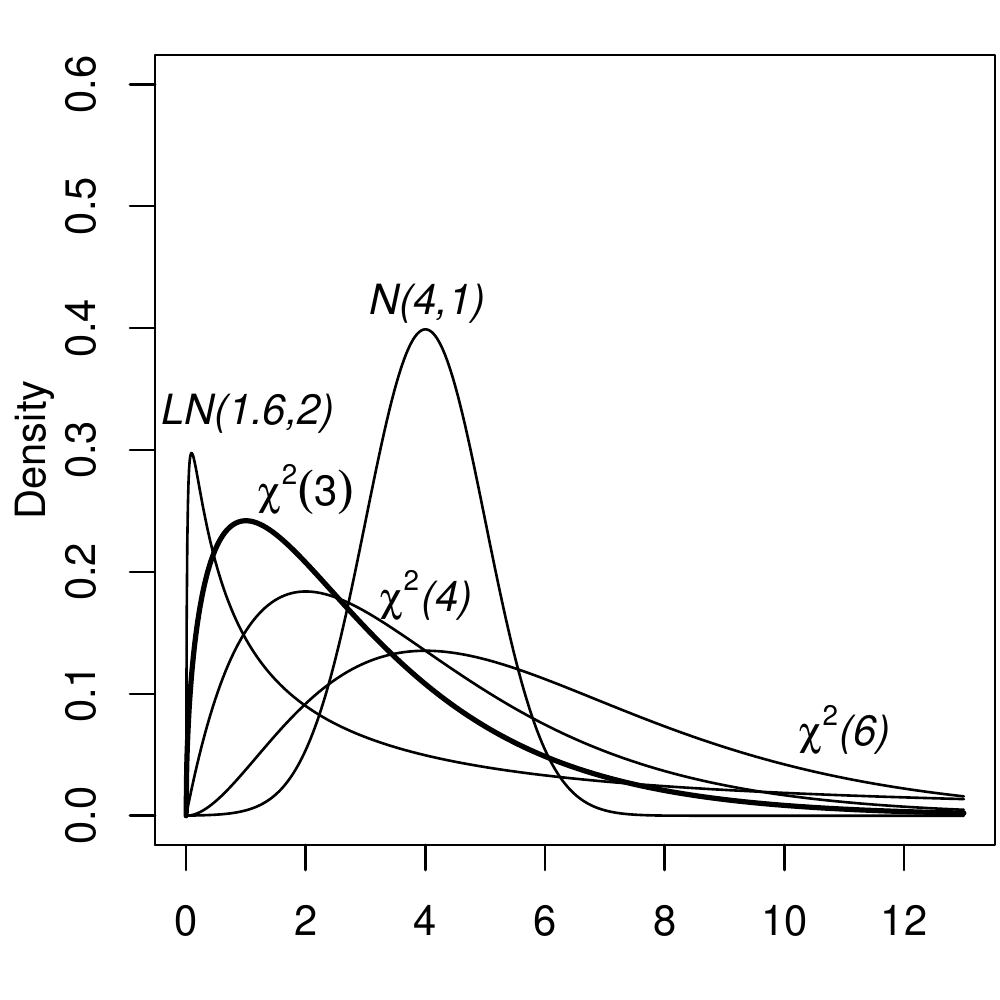}}
   \subfigure[Skewed densities corresponding to Subfigure (g) and (h) in Figure \ref{fig:power_TU}.]{\includegraphics[width=0.49\textwidth]{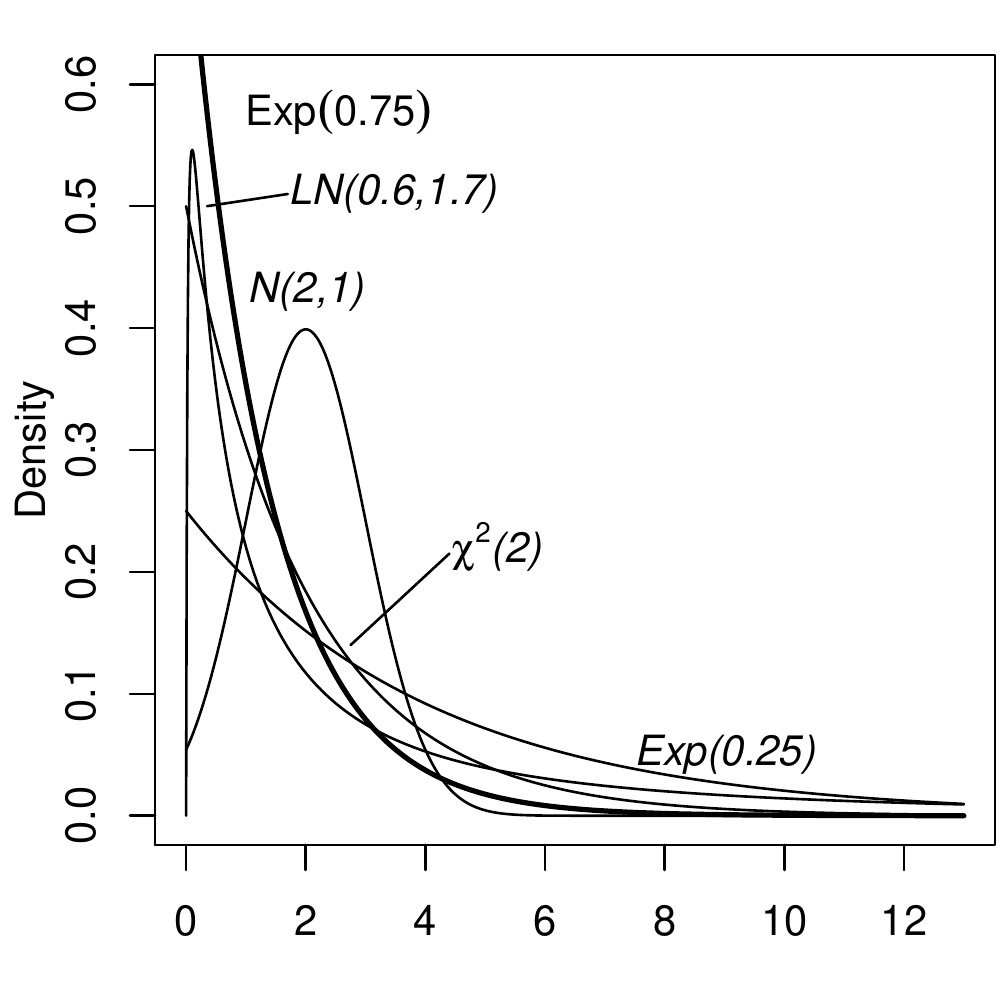}}
  \caption{Densities of the distributions used for the power calculations displayed in Figure \ref{fig:power_TU}. The density of $G$ is plotted in bold. Abbreviations: $\mathcal{N}(\mu,\sigma) =$ Normal distribution with mean $\mu$ and standard deviation $\sigma$; $\chi^2(df)=$ Chi-square distribution with $df$ degrees of freedom; $LN(\mu,\sigma) =$ Log-normal distribution with log-mean $\mu$ and log-standard deviation $\sigma$; $Exp(\lambda) =$ Exponential distribution with rate $\lambda$.} 
  \label{fig:dens}
\end{figure}

\begin{figure}[hbtp] 
  \subfigure[Varied $\sigma$: $\mathcal{N}(0.75,\sigma)$ vs. $\mathcal{N}(0,1)$. $D_U(0.5)$: $0-4\%$;  $D_{het}(0.5)$: $0-24\%$.]{\includegraphics[width=0.49\textwidth]{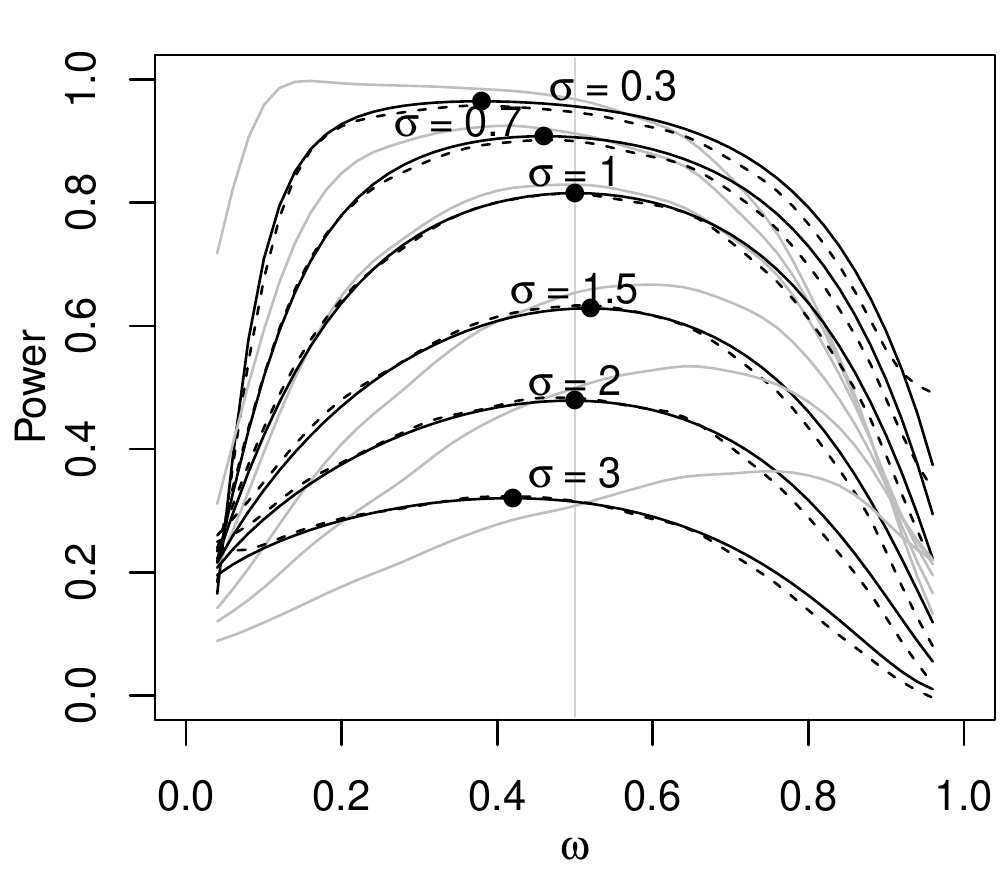}}
  \subfigure[Varied $N$: $\mathcal{N}(0.75,2)$ vs. $\mathcal{N}(0,1)$. $D_U(0.5)$: $0-2\%$;  $D_{het}(0.5)$: $0-12\%$.]{\includegraphics[width=0.49\textwidth]{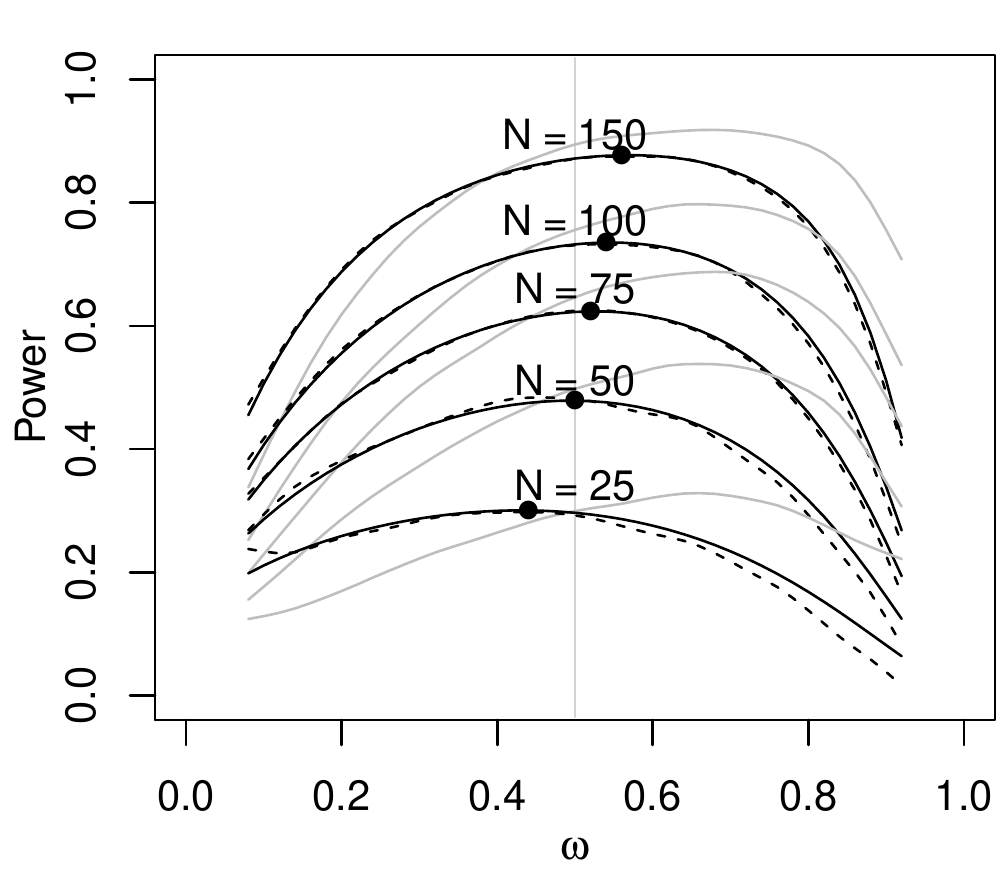}}
  \subfigure[Varied $\alpha$: $\mathcal{N}(0.75,0.7)$ vs. $\mathcal{N}(0,1)$. $D_U(0.5)$: $0-2\%$; $D_{het}(0.5)$: $0-10\%$.]{\includegraphics[width=0.49\textwidth]{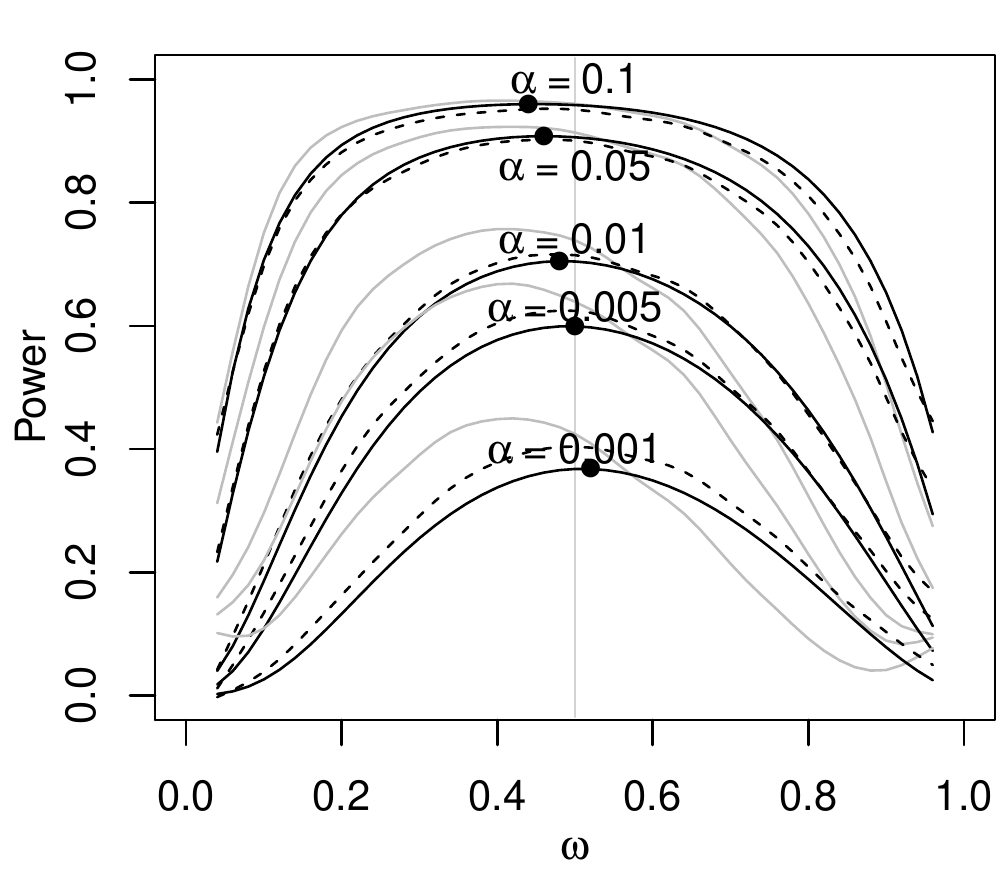}}
  \subfigure[Varied $F$ and $a$: $G(x) = F(x + a)$. $D_U(0.5)$: $0-2\%$.]{\includegraphics[width=0.49\textwidth]{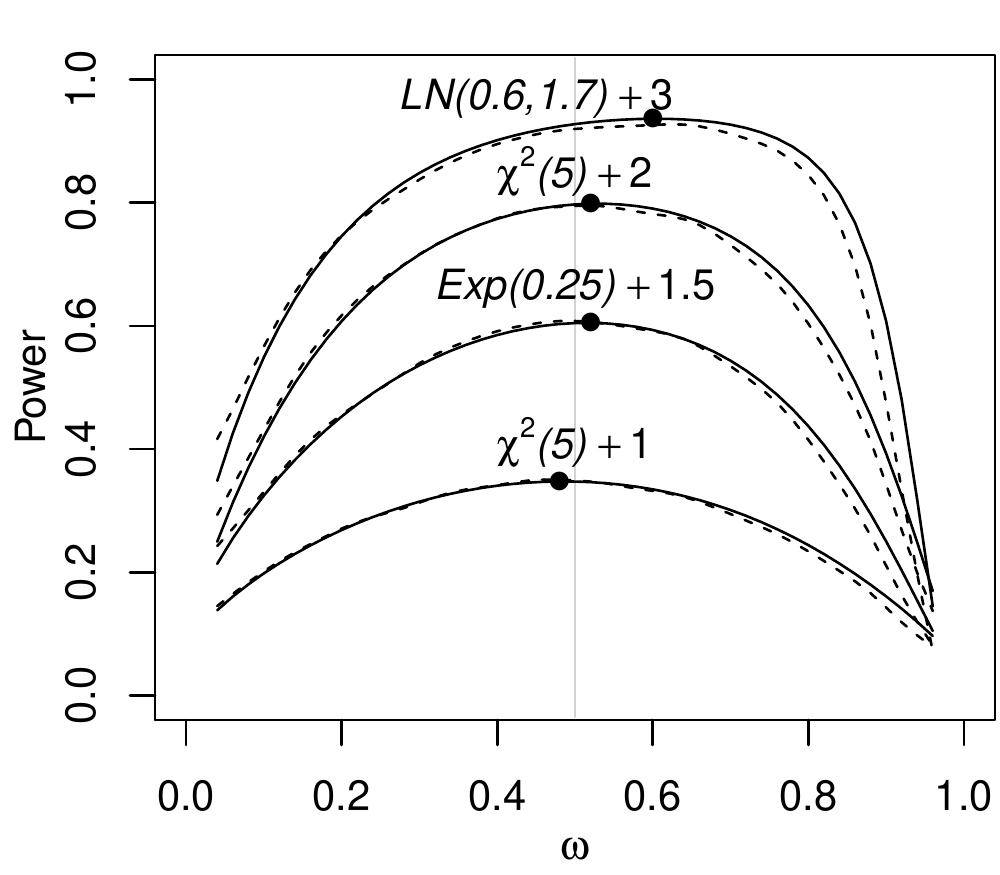}}
  \subfigure[Varied $N$: $\chi^2(5) + 1.5$ vs. $\chi^2(5)$.  $D_U(0.5)$: $0-2\%$.]{\includegraphics[width=0.49\textwidth]{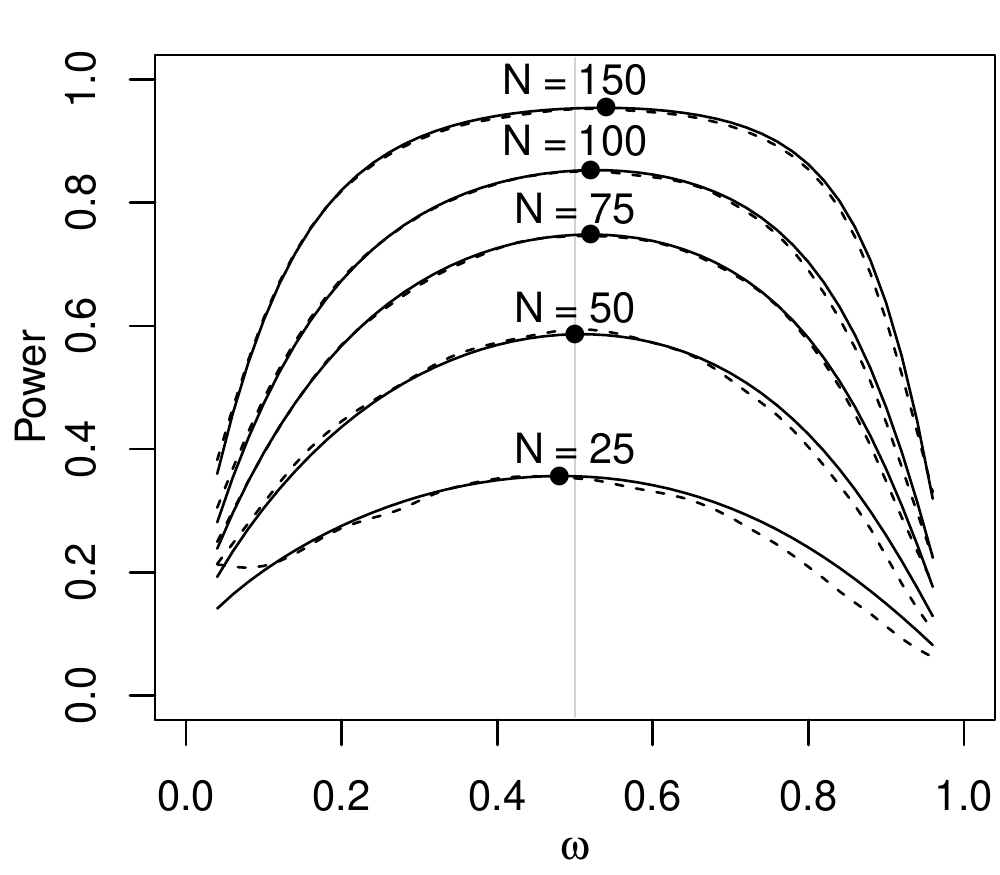}}
  \subfigure[Varied $F$: $G = \chi^2(3)$. $D_U(0.5)$: $0-2\%$.]{\includegraphics[width=0.49\textwidth]{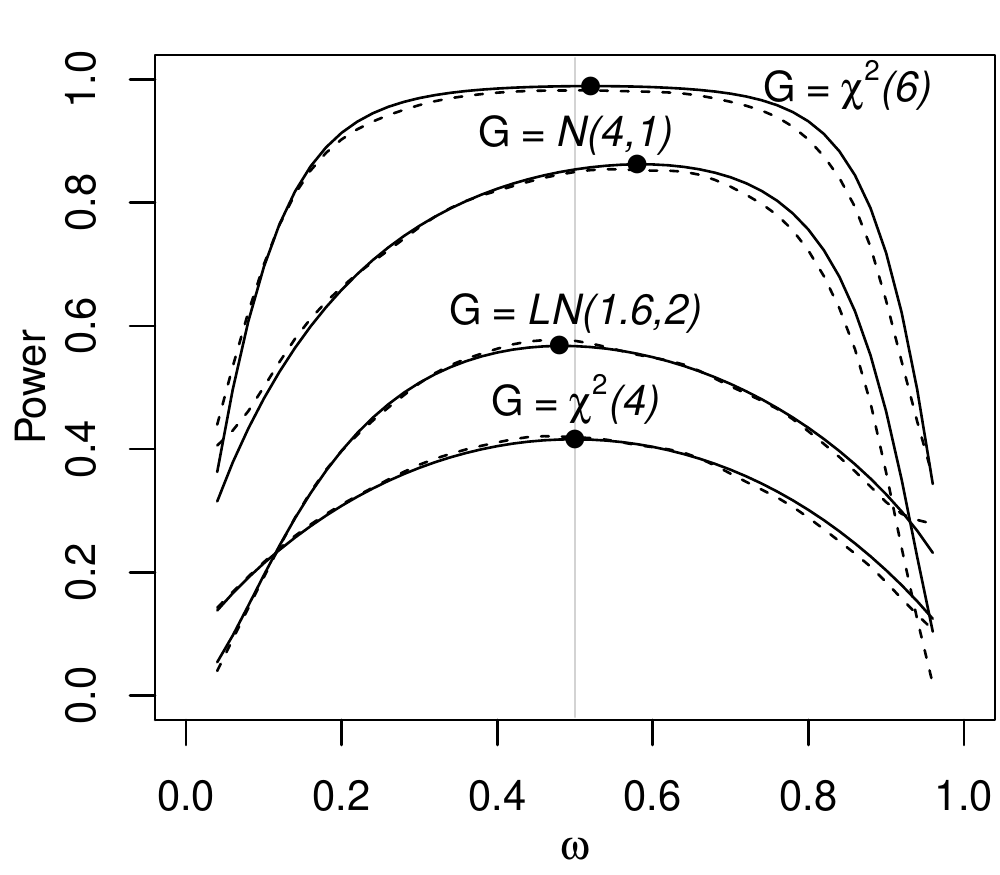}}
\end{figure}
\begin{figure}[hbtp]
 \subfigure[Varied $F$: $G = Exp(0.75)$.  $D_U(0.5)$: $0-2\%$.]{\includegraphics[width=0.49\textwidth]{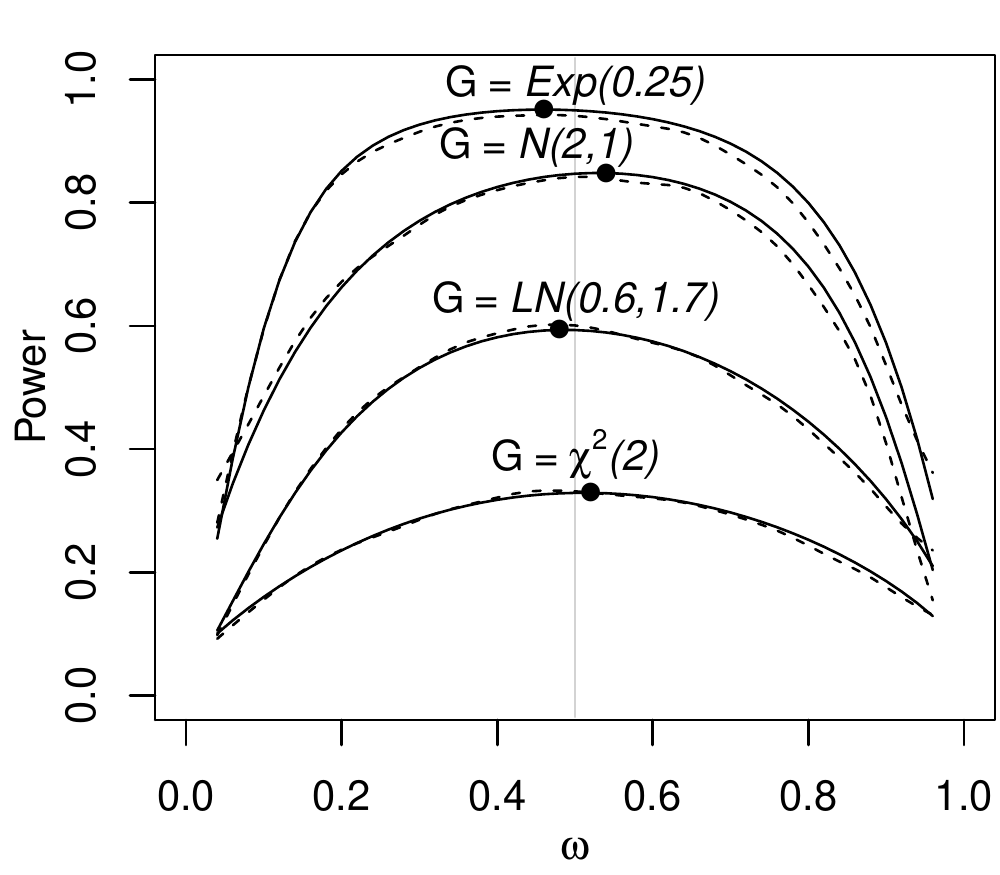}}
  \subfigure[Varied $\alpha$: $Exp(0.25)$ vs. $Exp(0.75)$.  $D_U(0.5)$: $0-2\%$.]{\includegraphics[width=0.49\textwidth]{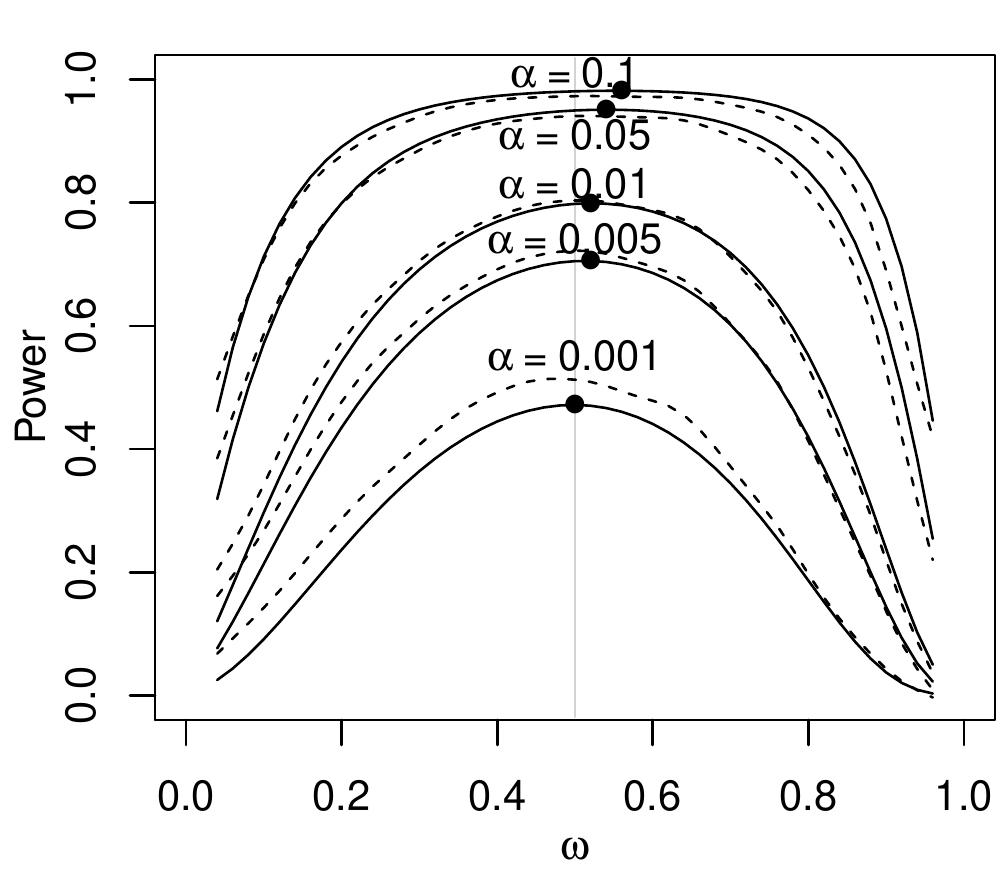}} 
   \subfigure{\includegraphics[width=0.99\textwidth]{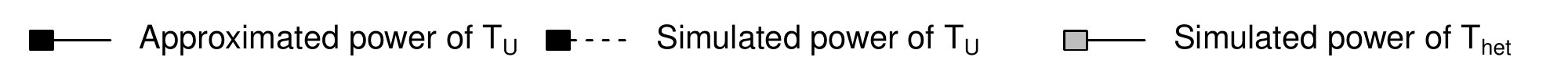}} 
   \subfigure{\includegraphics[width=0.99\textwidth]{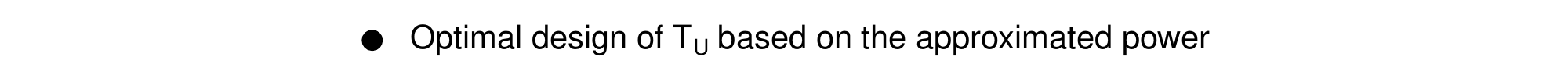}}  
  
  \caption{Power of $T_U$ for $N=50$ and $\alpha=0.05$ unless otherwise specified. Vertical grey lines indicate the position of the design $\omega = 0.5$. $D_U(0.5)$ and $D_{het}(0.5)$ indicate the deficiency of $\omega = 0.5$ relative to the optimal design for $T_U$ and $T_{het}$ respectively.  Abbreviations: $\mathcal{N}(\mu,\sigma) =$ Normal distribution with mean $\mu$ and standard deviation $\sigma$; $\chi^2(df)=$ Chi-square distribution with $df$ degrees of freedom; $LN(\mu,\sigma) =$ Log-normal distribution with log-mean $\mu$ and log-standard deviation $\sigma$; $Exp(\lambda) =$ Exponential distribution with rate $\lambda$.} 
	\label{fig:power_TU}
\end{figure}

\begin{figure}[hbtp] 
  \centering
  \includegraphics[height=0.50\textheight]{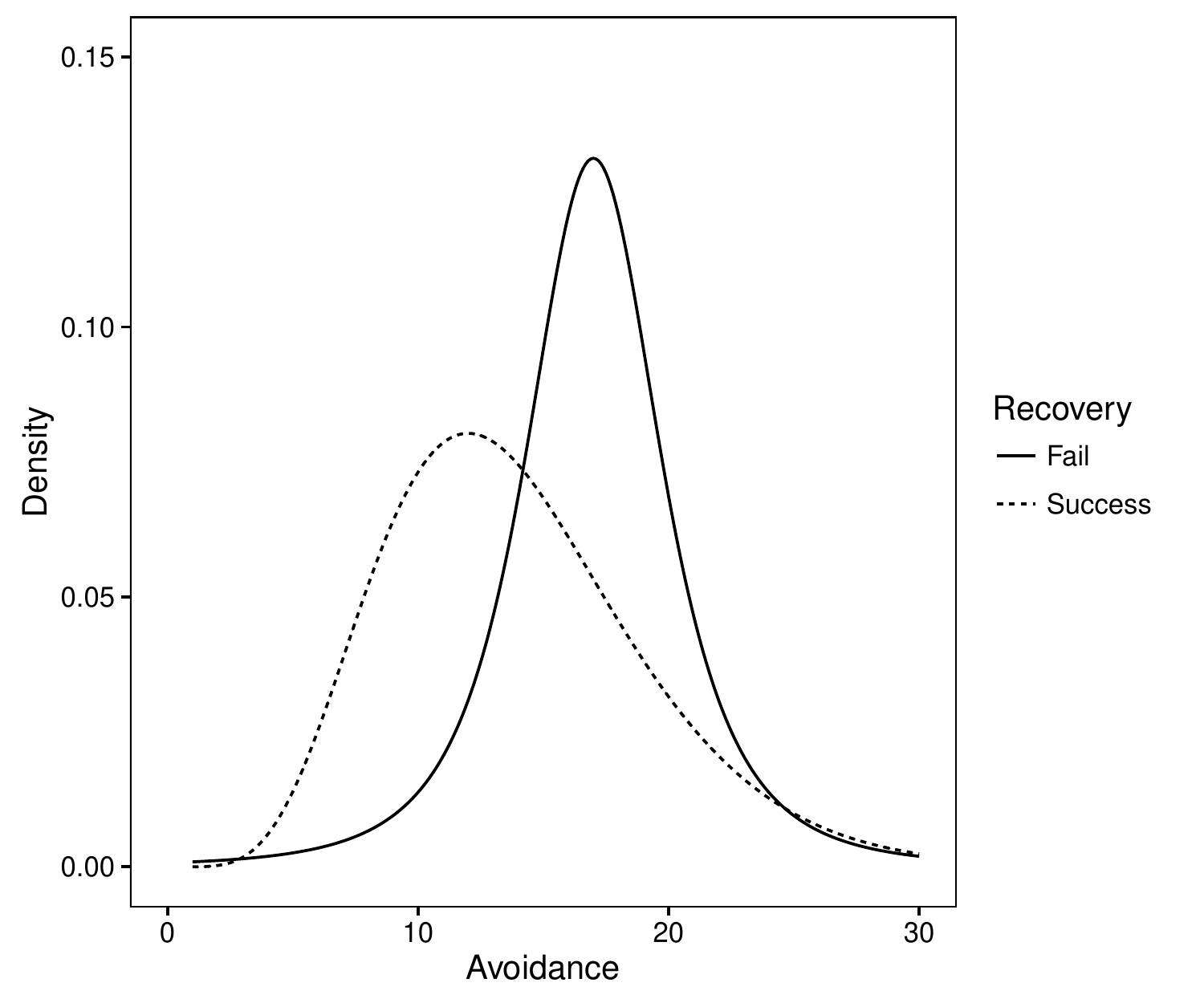}
  \caption{Approximated densities of the distribution of cognitive avoidance for patients who managed to recover from cancer (Success) and for patients who did not recover (Fail) in the study of \cite{epping1994}. Details are provided in Section~\ref{section_general_hyp}.}
  \label{fig:example_approx_dens}
\end{figure}

\begin{figure}[hbtp] 
  \centering
  \subfigure{\includegraphics[height=0.50\textheight]{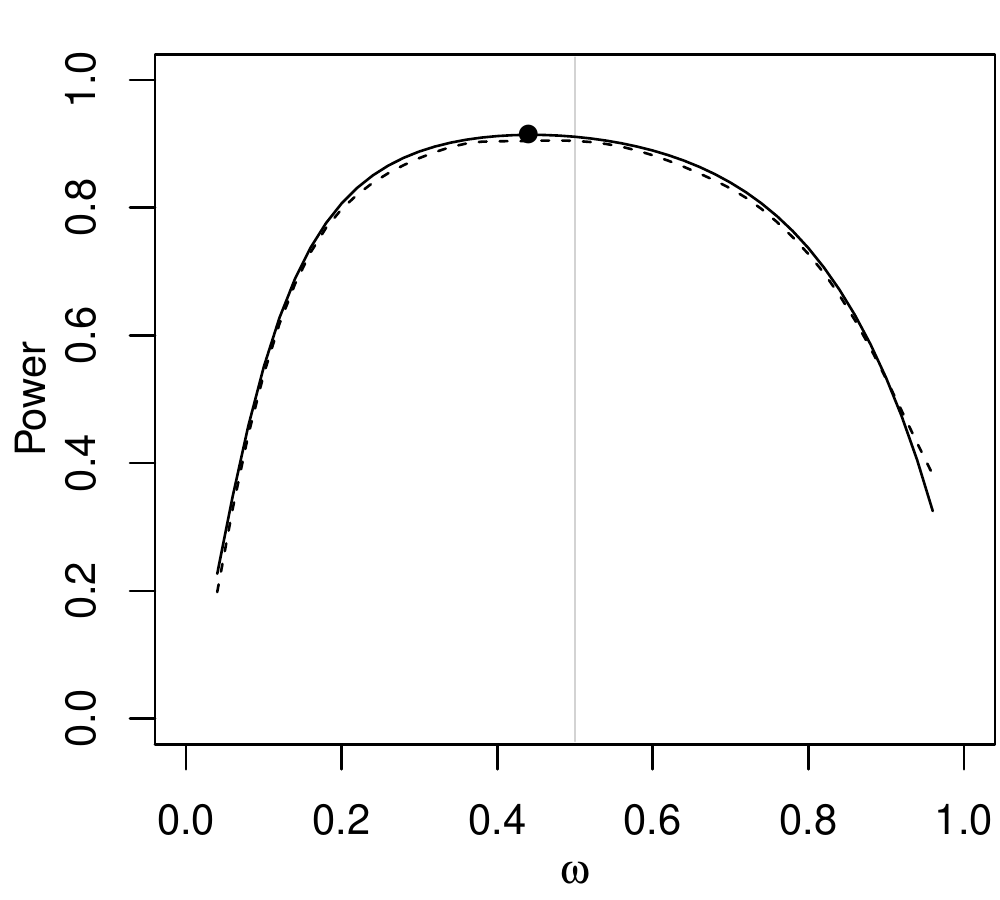}}
  \subfigure{\includegraphics[width=0.99\textwidth]{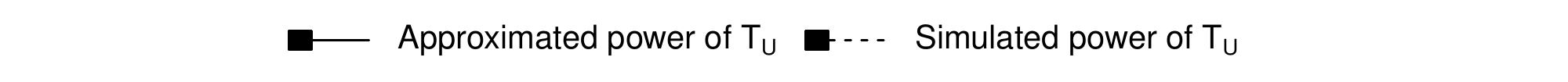}} 
   \subfigure{\includegraphics[width=0.99\textwidth]{Grafiken/legend_dots.pdf}}
  \caption{Power of $T_U$ applied to the study of \cite{epping1994}. The vertical grey line indicates the position of the design $\omega = 0.5$. Deficiency of $\omega = 0.5$ relative to the optimal design: $1\%$. More details are provided in Section~\ref{section_general_hyp}.}
  \label{fig:example_power}
\end{figure}

\end{document}